\documentclass[11pt]{article}

\usepackage[sort,numbers]{natbib}
\usepackage[bookmarks=true,bookmarksnumbered=true,colorlinks=true,allcolors=cyan]{hyperref}

\usepackage{amsmath,amssymb,amsfonts,amsthm}
\usepackage{lscape}
\usepackage{arydshln}
\renewcommand*{\baselinestretch}{1.25}

\usepackage{geometry}
\usepackage{indentfirst}

\usepackage{enumitem,color}

\newtheorem{theorem}{Theorem}[section]
\newtheorem{lemma}{Lemma}[section]
\newtheorem{proposition}{Proposition}[section]
\newtheorem{corollary}{Corollary}[section]
\theoremstyle{definition}
\newtheorem{definition}{Definition}[section]
\newtheorem*{rmk*}{Remark}
\newtheorem{rmk}{Remark}[section]

\DeclareMathOperator{\supp}{supp}
\DeclareMathOperator*{\esssupp}{ess\,supp}
\DeclareMathOperator{\sign}{sign}

\allowdisplaybreaks

\numberwithin{equation}{section}

\makeatletter
    \renewcommand*{\section}{\@startsection{section}{1}{\z@}%
    {10pt}{5pt}{\reset@font\normalsize\bfseries}}
\makeatother
    
\makeatletter
    \renewcommand*{\subsection}{\@startsection{subsection}{2}{\z@}%
    {5pt}{5pt}{\reset@font\normalsize\mdseries\itshape}}
\makeatother

\makeatletter
    \renewcommand*{\subsubsection}{\@startsection{subsubsection}{3}{\z@}%
    {5pt}{5pt}{\reset@font\normalsize\mdseries\itshape}}
\makeatother

\makeatletter
\def\@listi{\leftmargin\leftmargini
  \topsep=.5\baselineskip 
  \partopsep=0pt \parsep=0pt \itemsep=0pt}
\let\@listI\@listi
\@listi
\def\@listii{\leftmargin\leftmarginii
  \labelwidth\leftmarginii \advance\labelwidth-\labelsep
  \topsep=0pt \partopsep=0pt \parsep=0pt \itemsep=0pt}
\def\@listiii{\leftmargin\leftmarginiii
  \labelwidth\leftmarginiii \advance\labelwidth-\labelsep
  \topsep=0pt \partopsep=0pt \parsep=0pt \itemsep=0pt}
\def\@listiv{\leftmargin\leftmarginiv
  \labelwidth\leftmarginiv \advance\labelwidth-\labelsep
  \topsep=0pt \partopsep=0pt \parsep=0pt \itemsep=0pt}
\makeatother

\title{No arbitrage and lead-lag relationships}
\author{Takaki Hayashi\thanks{Keio University, Graduate School of Business Administration, 4-1-1 Hiyoshi, Yokohama 223-8526, Japan}
\thanks{Department of Business Administration, Graduate School of Social Sciences, Tokyo Metropolitan University, Marunouchi Eiraku Bldg. 18F, 1-4-1 Marunouchi, Chiyoda-ku, Tokyo 100-0005 Japan}
\thanks{CREST, Japan Science and Technology Agency}
\and
Yuta Koike
\thanks{Graduate School of Mathematical Sciences, The University of Tokyo, 3-8-1 Komaba, Meguro-ku, Tokyo 153-8914 Japan}
\footnotemark[2]
\thanks{The Institute of Statistical Mathematics, 10-3 Midori-cho, Tachikawa, Tokyo 190-8562, Japan}
\footnotemark[3]}

\begin{document}

\maketitle

\begin{abstract}

The existence of time-lagged cross-correlations between the returns of a pair of assets, which is known as the lead-lag relationship, is a well-known stylized fact in financial econometrics. Recently some continuous-time models have been proposed to take account of the lead-lag relationship. Such a model does not follow a semimartingale as long as the lead-lag relationship is present, so it admits an arbitrage without market frictions. 
In this paper we show that they are free of arbitrage if we take account of market frictions such as the presence of minimal waiting time on subsequent transactions or transaction costs. 

\vspace{3mm}


\end{abstract}

\section{Introduction}

In finance absence of arbitrage is considered as one of the intrinsic properties of financial markets. 
For an idealized market, the absence of arbitrage is characterized by the existence of an equivalent martingale measure (see e.g.~\citet{DS1994}). As a consequence, in such a market the price processes \textit{necessarily} follow a semimartingale under the no arbitrage assumption because the semimartingale property is invariant under change of measures.  

Nevertheless, empirical work occasionally suggests some evidence conflicting this conclusion. The existence of cross-correlations across returns of multiple assets, which is known as the \textit{lead-lag relationship}, is one of the representative phenomena among such evidence. Empirical investigation of lead-lag relationships in financial markets has drawn attention in financial econometrics for a long time. 
A notable example is the lead-lag relationship between a stock index and the index futures, where many authors have reported that the futures lead the index (see e.g.~\cite{KKK1987,SW1990,Chan1992,deJN1997,HA2014}). 
More generally, lead-lag relationships between an asset and its derivatives have been examined in several articles such as \cite{deJD1998,BOW2017}. 
Also, it is commonly observed that large firm returns tend to lead small firm returns (cf.~\cite{LM1990llag,CSS2011,PT2004}). 
In addition, \citet{Reno2003} have pointed out that lead-lag relationships play a key role to explain the \textit{Epps effect}, another well-known empirical fact named after \citet{Epps1979}.   

One common way to fill in the gap between the theoretical implication and the empirical suggestion as above is to take various sorts of market frictions into account. Discreteness of transaction times is one of the primal sources of such frictions which are well-investigated in the literature.  
A pioneer work in this area is \citet{Cheridito2003}, which has shown that the (geometric) fractional Brownian motion (fBm) model admits no arbitrage if we impose a (constant) minimal waiting time on subsequent transactions. Since the fBm is not a semimartingale unless its Hurst index is $1/2$, this result is out of the scope of the traditional characterization of no arbitrage for frictionless markets. The class of trading strategies with this type of restriction is named the ``Cheridito class'' by \citet{JPS2009}, and they have provided a characterization for a market consisting one risk-free asset and one risky asset to admit no arbitrage in the Cheridito class (see Lemma 1 of \cite{JPS2009}). This result has been extended to the case of multiple risky assets in \citet{Sayit2013}. There are also some articles attempting to relax the requirement of imposing a minimal waiting time on subsequent transactions. For example, \citet{BSV2011} have introduced a new class of strategies called the \textit{delay-simple strategies}, which, roughly speaking, allow the minimal waiting time on subsequent transactions to be random to some extent. They have given a characterization for a market having no arbitrage within the class of delay-simple strategies. More broadly, \citet{Bender2012} has admitted \textit{all} simple strategies and given a characterization of no-arbitrage within such a class. 

Another dominant source of market frictions is the existence of transaction costs. 
For discrete-time models, an adequate characterization for absence of arbitrage under transaction costs have been established in \citet{KS2001}, \citet{KRS2002} and \citet{Schachermayer2004}. 
For continuous-time models, \citet{Guasoni2006} has provided a sufficient condition for a market being free of arbitrage under constant proportional transaction costs. As a special case, it has shown that the geometric fBm model has no arbitrage under transaction costs. The result of \cite{Guasoni2006} is further refined by several authors such as \cite{GRS2008,GRS2010,SV2011,BPS2015,GLR2012}. 
In particular, \citet{GRS2010} have proven a version of the fundamental theorem of asset pricing under constant proportional transaction costs for the continuous and one-risky asset case. \citet{GLR2012} have extended this result to the situation that the price process is possibly discontinuous and the transaction costs are not constant.    

To our knowledge, however, there is no work which studies the no arbitrage property of a market model with a lead-lag relationship under market frictions. The aim of this paper is to shed a light on this issue, and we especially focus on continuous-time models. Rather recently, \citet{HRY2013} have proposed a lead-lag model in continuous-time (see also \citet{RR2010}), which is based on Brownian motion driven modeling and contains traditional It\^o processes as a special case, hence it is readily compatible with the traditional mathematical finance theory. 
Related models have been subsequently studied by several authors such as \cite{HA2014,AM2014,CCI2016,BOW2017} for empirical work and \cite{HK2016,HK2017,BCP2017,Koike2017laglan,Koike2017sllag,Chiba2017} from a statistical point of view, but there is no work in the context of mathematical finance. We intend to bridge the gap between those two areas in this work.    

This paper is organized as follows. 
Section \ref{sec:review} briefly reviews some no arbitrage results under market frictions in the existing literature. 
Section \ref{sec:main} presents main results obtained in this paper. 
All the proofs are collected in Section \ref{sec:proofs}. 

\section*{Notation}

For $x=(x^1,\dots,x^d)\in\mathbb{R}^d$, we denote by $\|x\|$ the Euclid norm of $x$, i.e.~$\|x\|^2=\sum_{i=1}^d(x^i)^2$. 
$\mathrm{Leb}$ denotes the Lebesgue measure on $\mathbb{R}$. 
For a topological space $\Xi$, we write $\mathcal{B}(\Xi)$ the Borel $\sigma$-field of $\Xi$. 
Given a $d$-dimensional process $X=(X_t)_{t\in[0,T]}$, we denote the $i$-th component process of $X$ by $X^i=(X^i_t)_{t\in[0,T]}$ for every $i=1,\dots,d$. 

\section{No arbitrage under market frictions: A review}\label{sec:review}


We consider a discounted market with one risk-free asset and $d$ risky assets traded on a finite time horizon $[0,T]$. The risk-free asset is used as a num\'eraire and thus assumed to be constantly equal to one. The risky assets are modeled by a $d$-dimensional stochastic process $S=(S_t)_{t\in[0,T]}$ defined on a filtered probability space $(\Omega,\mathcal{F},\mathbb{F}=(\mathcal{F}_t)_{t\in[0,T]},P)$ satisfying the usual conditions of completeness and right-continuity of the filtration. 
We assume that $S$ is c\`adl\`ag and adapted to $\mathbb{F}$. 

\subsection{Trading with simple strategies}

First we review some no arbitrage results when we restrict the class of admissible strategies to discrete trading. 
We mention that Sections 3.3--3.4 of \citet{BSV2011} have provided an excellent survey (and some original results) on this topic for the univariate case with an emphasis on models driven by (mixed) fractional Brownian motion. 
\begin{definition}
(a) A $d$-dimensional process $\Phi=(\Phi_t)_{t\in[0,T]}$ is called a \textit{simple strategy} (with respect to $\mathbb{F}$) if it is of the form
\begin{equation}\label{simple}
\Phi_t=\phi_01_{\{0\}}(t)+\sum_{j=0}^{n-1}\phi_j1_{(\tau_j,\tau_{j+1}]}(t),\qquad t\in[0,T],
\end{equation}
where $n\in\mathbb{N}$, $0=\tau_0\leq\tau_1\leq\cdots\leq\tau_n$ are a.s.~finite $\mathbb{F}$-stopping times, and $\phi_j$ is a $d$-dimensional $\mathcal{F}_{\tau_j}$-measurable random vector for every $j=0,1,\dots,n$. We denote by $\mathcal{S}(\mathbb{F})$ the set of all simple strategies with respect to $\mathbb{F}$. 

\noindent(b) A simple strategy $\Phi$ of the form \eqref{simple} is said to belong to the \textit{Cheridito class} (with respect to $\mathbb{F}$) if there exists a constant $h>0$ such that $\tau_{j+1}-\tau_j\geq h$ for all $j=0,1,\dots,n-1$. We denote by $\mathcal{L}(\mathbb{F})$ the set of all simple strategies belonging to the Cheridito class with respect to $\mathbb{F}$. 
\end{definition}

\begin{rmk}
The name ``Cheridito class'' is owed to \citet{JPS2009}, and it comes from the seminal work of \citet{Cheridito2003}: See Remark \ref{rmk:fBm}.
\end{rmk}

Throughout the paper, we only consider self-financing strategies. Since the market is already discounted, this means that the value process of a simple strategy $\Phi$ of the form \eqref{simple} with initial capital $v\in\mathbb{R}$ is given by
\[
V_t(\Phi;v)=v+\sum_{j=1}^{n-1}\phi_j^\top(S_{t\wedge\tau_{j+1}}-S_{t\wedge\tau_j}),\qquad t\in[0,T].
\]

\begin{definition}
A simple strategy $\Phi$ is called an \textit{arbitrage} if $P(V_T(\Phi;0)\geq0)=1$ and $P(V_T(\Phi;0)>0)>0$.
\end{definition}

\begin{rmk}\label{rmk:fBm}
\citet{Cheridito2003} has shown that there is no arbitrage in the Cheridito class for the (geometric) fractional Brownian motion model.
\end{rmk}

Now we present a sufficient condition for a market being free of arbitrage in the Cheridito class, which is investigated in \citet{Sayit2013}. Let $X=(X_t)_{t\in[0,T]}$ be a $d$-dimensional c\`adl\`ag $\mathbb{F}$-adapted process. For any two $\mathbb{F}$-stopping times $\tau_1,\tau_2$ such that $\tau_1\leq\tau_2$ a.s., we set
\[
A_i^+=\{X^i_{\tau_1}<X^i_{\tau_2}\},\qquad
A_i^-=\{X^i_{\tau_1}>X^i_{\tau_2}\},\qquad,
i=1,\dots,d.
\]
\begin{definition}
We say that $X$ satisfies the \textit{joint $\mathbb{F}$-conditional up and down (CUD) condition with respect to $\mathcal{L}(\mathbb{F})$} if for any $h\in(0,T)$ and any two stopping times $\tau_1\leq\tau_2$ with $\tau_2\geq\tau_1+h$ a.s., and any $B\in\mathcal{F}_{\tau_1}$ with $P(B)>0$, the following holds
\[
P\left(\left[\bigcap_{i\in I^B}A_i^{\alpha_i}\right]\cap B\right)>0,
\]
whenever $I^B\neq\emptyset$, where $\alpha_1,\dots,\alpha_d\in\{+,-\}$ and $I^B$ is the set of all $k\in\{1,\dots,d\}$ such that $P(\{X_{\tau_1}^k\neq X_{\tau_2}^k\}\cap B)>0$. 
\end{definition}
\begin{rmk}
The above condition is first considered in \citet{JPS2009} for the univariate case. The name ``CUD condition'' is owed to \citet{BSV2011}. 
\end{rmk}
\citet{Sayit2013} have shown that this condition is sufficient for the absence of arbitrage in the Cheridito class:
\begin{proposition}[\cite{Sayit2013}, Proposition 1]\label{prop:sayit}
If $S$ satisfies the joint $\mathbb{F}$-CUD condition with respect to $\mathcal{L}(\mathbb{F})$, then there is no arbitrage in the Cheridito class. 
\end{proposition}
\begin{rmk}
In the case of $d=1$, the converse of Proposition \ref{prop:sayit} is also true by Lemma 1 of \cite{JPS2009}. However, when $d>1$, the converse of Proposition \ref{prop:sayit} does not necessarily hold; see page 617 of \cite{Sayit2013}.
\end{rmk}
We note that the joint $\mathbb{F}$-CUD condition is invariant under component-wise transform by strictly monotone functions:
\begin{proposition}[\cite{Sayit2013}, Proposition 2]\label{prop:cud-trans}
For each $i=1,\dots,d$, let $f_i:\mathbb{R}\to\mathbb{R}$ be a strictly monotone function. $S$ satisfies the joint $\mathbb{F}$-CUD condition with respect to $\mathcal{L}(\mathbb{F})$ if and only if the $d$-dimensional process $(f_1(S^1_t),\dots,f_d(S^d_t))$ $(t\in[0,T])$ satisfies the joint $\mathbb{F}$-CUD condition with respect to $\mathcal{L}(\mathbb{F})$. 
\end{proposition}

Finally, we make some comments on no arbitrage results for simple strategies.
\begin{rmk}
\citet{Bender2012} have given a necessary and sufficient condition for the market $S$ being free of arbitrage in the class $\mathcal{S}(\mathbb{F})$ of all simple strategies when $S$ is continuous. Moreover, using \citet{Bender2012}'s result, \citet{Peyre2017} has shown that there is no arbitrage in the class $\mathcal{S}(\mathbb{F})$ for the fractional Brownian motion market.
\end{rmk}

\subsection{Trading under transaction costs}

Next we review some no arbitrage results when we take account of transaction costs in trading. Here, we focus on constant proportional transaction costs and adopt the setting of \citet{Guasoni2002} (see also \cite{Guasoni2006,SV2011}). We refer to \cite{KS2009} for a comprehensive treatment of this topic. 

Throughout this subsection, we assume that $S$ is quasi-left continuous (cf.~Chapter I, Definition 2.25 of \cite{JS}). 
Let $\varepsilon>0$. For a $d$-dimensional left-continuous $\mathbb{F}$-adapted process $\Phi=(\Phi_t)_{t\in[0,T]}$ with finite variation, we define the \textit{value process of $\Phi$ with $\varepsilon$-transaction costs} (with zero initial capital) by
\[
V_t^\varepsilon(\Phi)=\sum_{i=1}^d\int_0^t\Phi^i_sdS^i_s-\sum_{i=1}^d\left(\varepsilon\int_0^tS^i_sd\mathrm{TV}(\Phi^i)_s+\varepsilon|\Phi^i_t|S^i_t\right),\qquad t\in[0,T].
\]
Here, for each $i=1,\dots,d$, $\mathrm{TV}(\Phi^i)=(\mathrm{TV}(\Phi^i)_t)_{t\in[0,T]}$ is the total variation process of $\Phi^i$, and the integral $\int_0^t\Phi^i_sdS^i_s$ is defined as in Definition 2.2 of \cite{Guasoni2002}. 
\begin{definition}
Let $\varepsilon>0$ and let $\Phi$ be a $d$-dimensional left-continuous $\mathbb{F}$-adapted process with finite variation. 
\begin{enumerate}[label=(\alph*)]

\item $\Phi$ is called an \textit{admissible strategy with $\varepsilon$-transaction costs} if there is a constant $M>0$ such that $V^\varepsilon_t(\Phi)\geq-M$ a.s.~for all $t\in[0,T]$. We write $\mathcal{A}^\varepsilon$ the class of all admissible strategies with $\varepsilon$-transaction costs.

\item $\Phi$ is called an \textit{arbitrage with $\varepsilon$-transaction costs} if 
$P(V^\varepsilon_T(\Phi)\geq0)=1$ and $P(V^\varepsilon_T(\Phi)>0)>0$.

\end{enumerate}

\end{definition}
\begin{rmk}
The left continuity of admissible strategies in the above definition can be relaxed to the predictability because for any $\mathbb{F}$-predictable process $\Phi=(\Phi_t)_{t\in[0,T]}$ with finite variation it holds that $V^\varepsilon_t(\Phi)=V^\varepsilon_t(\Phi_-)$ a.s.~for all $t\in[0,T]$ and $\varepsilon>0$ by Proposition 2.5 of \cite{Guasoni2002}, where $\Phi_-=(\Phi_{t-})_{t\in[0,T]}$.
\end{rmk}
Now we present a sufficient condition for $S$ having no arbitrage with $\varepsilon$-transaction costs in the class $\mathcal{A}^\varepsilon$ for any $\varepsilon>0$, which is called the \textit{stickiness}.
\begin{definition}\label{def:sticky}
A $d$-dimensional c\`adl\`ag $\mathbb{F}$-adapted process $X=(X_t)_{t\in[0,T]}$ is said to be \textit{sticky} (with respect to $\mathbb{F}$) if for any $t\in[0,T)$ and $\delta>0$,
\begin{equation*}
P\left(\bigcap_{i=1}^d\left\{\sup_{u\in[t,T]}|X_u^i-X^i_t|<\delta\right\}|\mathcal{F}_t\right)>0\qquad\text{a.s.}
\end{equation*}
\end{definition}
The above definition of the stickiness is due to Definition 2.2 of \citet{BPS2015}. However, as is pointed out in Remark 2.1 of \cite{BPS2015}, this definition turns out to be equivalent to the notion of \textit{joint stickiness} in \citet{SV2011}. More precisely, we have the following result:
\begin{lemma}[\cite{BPS2015}, Lemma 3.1]\label{lemma:sticky}
If a $d$-dimensional c\`adl\`ag $\mathbb{F}$-adapted process $X=(X_t)_{t\in[0,T]}$ is sticky, for any $\mathbb{F}$-stopping time $\tau:\Omega\to[0,T]$ and any $\mathcal{F}_\tau$-measurable non-negative random variable $\eta$, we have
\begin{equation*}
P\left(\bigcap_{i=1}^d\left\{\sup_{u\in[\tau,T]}|X_u^i-X^i_t|<\eta\right\}|\mathcal{F}_\tau\right)>0\qquad\text{a.s.~on }\{\eta>0\}.
\end{equation*}
\end{lemma} 
\begin{rmk}\label{rmk:sticky}
The stickiness was originally introduced in Definition 2.2 of \citet{Guasoni2006} for the univariate case. It has also been shown that fractional Brownian motion is sticky (Proposition 5.1 of \cite{Guasoni2006}).
\end{rmk}
Lemma \ref{lemma:sticky} and Proposition 2 of \cite{SV2011} imply that the stickiness is invariant under continuous transform:
\begin{proposition}\label{prop:sticky-trans}
Let $X=(X_t)_{t\in[0,T]}$ be a $d$-dimensional c\`adl\`ag $\mathbb{F}$-adapted process. Also, let $f:\mathbb{R}^d\to\mathbb{R}^d$ be a continuous function and define the $d$-dimensional process $Y=(Y_t)_{t\in[0,T]}$ by $Y_t=f(X^1_t,\dots,X^d_t)$, $t\in[0,T]$. If $X$ is sticky, then $Y$ is sticky as well.
\end{proposition}
Noting the above results, we obtain the following result from Proposition 1 of \citet{SV2011}:
\begin{proposition}
Suppose that $S^i_t>0$ and $S^i_{t-}>0$ for every $i=1,\dots,d$ and every $t\in[0,T]$. If $S$ is sticky, then $S$ has no arbitrage with $\varepsilon$-transaction costs in the class $\mathcal{A}^\varepsilon$ for all $\varepsilon>0$. 
\end{proposition}
Next we review the concept of \textit{consistent price system} (CPS), which plays an important role when we consider models with proportional transaction costs. 
\begin{definition}
Let $\varepsilon>0$. A $d$-dimensional $\mathbb{F}$-adapted process $M=(M_t)_{t\in[0,T]}$ is called an \textit{$\varepsilon$-consistent price system} ($\varepsilon$-CPS) for $S$ if the following conditions hold true:
\begin{enumerate}[label=(\roman*)]

\item We have
\[
\frac{S^i_t}{1+\varepsilon}\leq M^i_t\leq(1+\varepsilon)S^i_t\qquad\text{a.s.}
\]
for any $i\in\{1,\dots,d\}$ and $t\in[0,T]$.

\item There is a probability measure $Q$ on $(\Omega,\mathcal{F})$ equivalent to $P$ such that $M$ is a $d$-dimensional $\mathbb{F}$-martingale under $Q$. 

\end{enumerate}
\end{definition}
As is discussed in \citet{GRS2008,GRS2010}, the existence of CPS provides a useful tool to solve arbitrage and superreplication problems. As an illustration, we show that the existence of CPS implies no arbitrage in the num\'eraire-free sense.
\begin{definition}
Let $\varepsilon>0$. A $d$-dimensional left-continuous $\mathbb{F}$-adapted process $\Phi$ with finite variation is called an \textit{admissible strategy with $\varepsilon$-transaction costs in the num\'eraire-free sense} if there is a constant $M>0$ such that $V^\varepsilon_t(\Phi)\geq-M(1+\sum_{i=1}^dS^i_t)$ a.s.~for all $t\in[0,T]$. We write $\tilde{\mathcal{A}}^\varepsilon$ the class of all admissible strategies with $\varepsilon$-transaction costs in the num\'eraire-free sense.
\end{definition}
See Remark 2.17 of \cite{GRS2008} for a discussion of this definition of admissible strategies. 
The following result follows from Lemma 2.1 of \cite{Guasoni2006} and Theorem 2.6 of \cite{Yan1998}:
\begin{proposition}\label{prop:arb-tc2}
Suppose that $S^i_t>0$ for every $i=1,\dots,d$ and every $t\in[0,T]$. If $S$ has an $\varepsilon$-CPS for some $\varepsilon>0$, then $S$ has no arbitrage with $\varepsilon$-transaction costs in the class $\tilde{\mathcal{A}}^\varepsilon$.
\end{proposition}
\begin{rmk}
With a slightly different definition of the value process $V^\varepsilon(\Phi)$ and the admissible class $\tilde{\mathcal{A}}^\varepsilon$, \citet{GRS2010} have shown that the converse of Proposition \ref{prop:arb-tc2} also holds true in the case of $d=1$.  
\end{rmk}
Finally, we remark that \citet{BPS2015} have shown that the stickiness is sufficient for the existence of CPS as long as $S$ is continuous:
\begin{proposition}[\cite{BPS2015}, Theorem 2.1]\label{prop:cps}
Suppose that $S$ is continuous. If $S$ is sticky, then $S$ has an $\varepsilon$-CPS for any $\varepsilon>0$.
\end{proposition}
%

\subsection{Conditional full support property}

In the previous subsections we see that the joint $\mathbb{F}$-CUD condition (resp.~the stickiness) is sufficient for a market having no arbitrage in the Cheridito class (resp.~under proportional transaction costs). In this subsection we present a convenient property which implies both the joint $\mathbb{F}$-CUD condition and the stickiness.   

Throughout this subsection, we do not require that the filtered probability space $(\Omega,\mathcal{F},\mathbb{F},P)$ satisfies the usual hypotheses. For $-\infty<a< b<\infty$, we denote by $C([a,b],\mathbb{R}^d)$ the space of all continuous functions from $[a,b]$ to $\mathbb{R}^d$, equipped with the uniform topology. Also, we set $C_x([a,b],\mathbb{R}^d)=\{f\in C([a,b],\mathbb{R}^d):f(a)=x\}$ for $x\in\mathbb{R}^d$.

Let us recall the notion of \textit{support} of a probability measure defined on a metric space:
\begin{definition}
Let $\Xi$ be a separable metric space. For a probability measure $\mu$ on $(\Xi,\mathcal{B}(\Xi))$, the \textit{support} of $\mu$ is defined as the smallest closed set $C$ of $\Xi$ such that $\mu(C)=1$ (such a set $C$ always exists by Chapter II, Theorem 2.1 of \cite{Parth1967}). We denote by $\supp\mu$ the support of $\mu$. 
\end{definition}
Now we introduce the concept of \textit{conditional full support property}:
\begin{definition}
A $d$-dimensional continuous $\mathbb{F}$-adapted process $X=(X_t)_{t\in[0,T]}$ is said to have \textit{conditional full support} (CFS) with respect to $\mathbb{F}$ if
\[
\supp\mathcal{L}_P((X_t)_{t\in[t_0,T]}|\mathcal{F}_{t_0})=C_{X_{t_0}}([t_0,T],\mathbb{R}^d)\qquad\text{a.s.}
\]
for any $t_0\in[0,T)$, where $\mathcal{L}_P((X_t)_{t\in[t_0,T]}|\mathcal{F}_{t_0})$ denotes the regular conditional law of $(X_t)_{t\in[t_0,T]}$ on $C([t_0,T],\mathbb{R}^d)$ under $P$, given $\mathcal{F}_{t_0}$. 
\end{definition}
We list some processes having CFS (under some reasonable assumptions) in Table \ref{table:cfs}. 
As remarked in \cite{PSY2017} (see Remark 2.4(ii) of \cite{PSY2017}), the CFS property is equivalent to the so-called \textit{conditional small ball property}:
\begin{lemma}\label{lemma:csbp}
For a $d$-dimensional continuous $\mathbb{F}$-adapted process $X=(X_t)_{t\in[0,T]}$, the following two conditions are equivalent: 
\begin{enumerate}[label=(\roman*)]

\item $X$ has CFS with respect to $\mathbb{F}$.

\item For any $t_0\in[0,T)$, $f\in C_0([t_0,T],\mathbb{R}^d)$ and $\varepsilon>0$,
\[
P\left(\sup_{t\in[t_0,T]}\|X_t-X_{t_0}-f(t)\|<\varepsilon|\mathcal{F}_{t_0}\right)>0\qquad\text{a.s.}
\]

\end{enumerate}
\end{lemma}

Now, Proposition 3 of \cite{Sayit2013} and Remark 2.2 of \cite{BPS2015} yield the following result:
\begin{proposition}\label{prop:cfs}
If a $d$-dimensional continuous $\mathbb{F}$-adapted process $X=(X_t)_{t\in[0,T]}$ have CFS with respect to $\mathbb{F}$, then $X$ satisfies the joint $\mathbb{F}$-CUD condition and is sticky with respect to $\mathbb{F}$.
\end{proposition}
To conclude this section, we enumerate some useful results on the CFS property. 
The first result is a direct consequence of the above lemma and Lemma 2.2 of \cite{Pakkanen2010}:
\begin{lemma}\label{lemma:large}
Let $X=(X_t)_{t\in[0,T]}$ be a $d$-dimensional continuous $\mathbb{F}$-adapted process. Let $\mathbb{G}=(\mathcal{G}_t)_{t\in[0,T]}$ be a filtration of $\mathcal{F}$ such that $\mathcal{F}_t\subset\mathcal{G}_t$ for all $t\in[0,T]$. Then, $X$ has CFS with respect to $\mathbb{F}$ if it has CFS with respect to $\mathbb{G}$. 
\end{lemma}
The next one is a multivariate extension of Lemma 2.3 from \cite{Pakkanen2010}, which states that the CFS property is invariant under augmentation of the filtration in the usual way (see e.g.~page 45 of \cite{RY1999} for the definition of the usual augmentation of a filtration). The proof is an easy extension of the original one and we omit it.
\begin{lemma}\label{lemma:augmentation}
Let $X=(X_t)_{t\in[0,T]}$ be a $d$-dimensional continuous $\mathbb{F}$-adapted process. Then, $X$ has CFS with respect to $\mathbb{F}$ if and only if it has CFS with respect to the usual augmentation of $\mathbb{F}$. 
\end{lemma}
The third one is a straightforward multivariate extension of Lemma 3.1 from \cite{GSvZ2011}. For a $d$-dimensional process $X=(X_t)_{t\in[0,T]}$ we write $\mathbb{F}^X=(\mathcal{F}_t^X)_{t\in[0,T]}$ the natural filtration of $X$.  
\begin{lemma}\label{lemma:equivalence}
Let $X=(X_t)_{t\in[0,T]}$ and $Y=(Y_t)_{t\in[0,T]}$ be $d$-dimensional continuous processes, possibly defined on different probability spaces. If the laws of $X$ and $Y$ on $C([0,T],\mathbb{R}^d)$ are equivalent, then $X$ has CFS with respect to $\mathbb{F}^X$ if and only if $Y$ has CFS with respect to $\mathbb{F}^Y$. 
\end{lemma}
%
The last one is a straightforward multivariate extension of Lemma 3.2 from \cite{GSvZ2011} (see also Remark 2.4(iii) of \cite{PSY2017}):
\begin{lemma}\label{lemma:sum}
Let $X=(X_t)_{t\in[0,T]}$ and $Y=(Y_t)_{t\in[0,T]}$ be mutually independent $d$-dimensional continuous $\mathbb{F}$-adapted processes. If $X$ has CFS with respect to $\mathbb{F}^X$, then $X+Y$ also has CFS with respect to its natural filtration.
\end{lemma}

\begin{table}[ht]
\caption{Processes having CFS (under some reasonable assumptions)}
\label{table:cfs}
\begin{center}
\begin{tabular}{lll}\hline
Process & & Source \\ \hline
\multicolumn{3}{c}{Univariate processes} \\
Fractional Brownian motion & & \citet{GRS2008}, Proposition 4.2 \\
Integrated process & & \citet{GRS2008}, Lemma 4.5 \\
Brownian moving average & & \citet{Cherny2008}, Theorem 1.1 \\
It\^o process & & \citet{Pakkanen2010} \\
Brownian semistationary process & & \citet{Pakkanen2011}, Corollary 3.1 \\
Gaussian process with stationary increments & & \citet{GSvZ2011}, Theorem 2.1 \\
\multicolumn{3}{c}{Multivariate processes} \\
Diffusion process & & \citet{GRS2008}, Example 4.1 \\ 
Independent fBms with possibly different Hurst parameters & & \citet{SV2011}, Proposition 3 \\ 
Multi-dimensional It\^o process & & \citet{HPR2014}, Theorem 2 \\
Multivariate Brownian moving average & & \citet{PSY2017}, Theorem 2.7 \\ \hline
\end{tabular}
\end{center}
\label{default}
\end{table}%

\section{No arbitrage with lead-lag relationships}\label{sec:main}

\subsection{Hoffmann-Rosenbaum-Yoshida model}

\citet{HRY2013} have proposed a novel continuous-time model for modeling lead-lag relationships. Roughly speaking, their model consists of one semimartingale and another ``delayed'' semimartingale. In the following we give a more precise description of a simplified version of their model which we focus on in this paper.  
 
Let $B^1=(B^1_t)_{t\in[0,\infty)}$ and $B^2=(B^2_t)_{t\in[0,\infty)}$ be two standard Brownian motions such that
\[
E[(B^1_t-B^1_s)(B^2_{t+\theta}-B^2_{s+\theta})]=\int_s^t\rho(u)du
\]
for $0\leq s<t<\infty$, where $\theta\geq0$ and $\rho:[0,\infty)\to[-1,1]$ is a deterministic function. Formally, such $B^1$ and $B^2$ can be constructed as follows. Let $W^k=(W^k_t)_{t\in[0,\infty)}$, $k=0,1,2,3$, be mutually independent standard Wiener processes. 
We define the processes $B^1$ and $B^2$ by
\begin{equation}\label{hry-real}
\left\{\begin{array}{l}
B^1_t=\int_0^t\sign(\rho(u))\sqrt{|\rho(u)|}dW^1_u+\int_0^{t}\sqrt{1-|\rho(u)|}dW^2_u,\\
B^2_t=W^0_{t\wedge\theta}+\int_0^{(t-\theta)_+}\sqrt{|\rho(u)|}dW^1_u+\int_0^{(t-\theta)_+}\sqrt{1-|\rho(u)|}dW^3_u
\end{array}\right.
\end{equation}
for $t\geq0$. It is not difficult to check that these $B^1$ and $B^2$ are the desired ones. 

Now, for each $\nu=1,2$, the (discounted) log-price process of the $\nu$-th risky asset is given by
\[
X^\nu_t=A^\nu_t+\int_0^t\sigma_\nu(u)dB^\nu_u,\qquad t\in[0,T],
\]
where $\sigma_\nu\in L^2(0,T)$ and $A^\nu=(A^\nu_t)_{t\in[0,T]}$ is a continuous process. In the original paper  \cite{HRY2013} the process $A^\nu$ is assumed to be of finite variation, but in this paper we will instead assume that $A^\nu$ is independent of the process $B$. 

Let us consider the bivariate process $X=(X^1,X^2)$. As the filtration $\mathbb{F}$, we consider the usual augmentation of $\mathbb{F}^X$. The following proposition establishes the CFS property of $X$ with respect to $\mathbb{F}$. 
\begin{proposition}\label{hry-cfs}
Suppose that the process $A=(A^1,A^2)$ is independent of $B$. Suppose also that $\mathrm{Leb}(\{t\in[0,T]:\sigma_\nu(t)=0\})=0$ for $\nu=1,2$ and $\mathrm{Leb}(\{t\in[0,T]:|\rho(t)|=1\})=0$. Then the process $X$ has CFS with respect to $\mathbb{F}$. 
\end{proposition}

Let us recall that, for each $\nu=1,2$, the (discounted) price process of $\nu$-th risky asset is given by $S^\nu_t=\exp(X^\nu_t)$, $t\in[0,T]$. Combining the above proposition with Propositions \ref{prop:sayit}--\ref{prop:sticky-trans} and \ref{prop:arb-tc2}--\ref{prop:cfs}, we obtain the following no-arbitrage properties of the Hoffmann-Rosenbaum-Yoshida model:
\begin{theorem}
Under the assumptions of Proposition \ref{hry-cfs}, the following statements hold true:
\begin{enumerate}[label=(\alph*)]

\item The market $S$ has no arbitrage in the Cheridito class.

\item The market $S$ has an $\varepsilon$-CPS for any $\varepsilon>0$. Consequently, $S$ has no arbitrage with $\varepsilon$-transaction costs in the class $\tilde{\mathcal{A}}^\varepsilon$ for any $\varepsilon>0$. 

\end{enumerate}
\end{theorem}

\subsection{Brownian motions with a general lead-lag structure}

In order to extend the Hoffmann-Rosenbaum-Yoshida model, \citet{HK2016} have investigated possible lead-lag structures for two Brownian motions, and they have obtained the following result:
\begin{proposition}[\cite{HK2016}, Proposition 2]\label{characterization}
Suppose that a measurable function $f:\mathbb{R}\to\mathbb{C}$ satisfies
\begin{equation}\label{L_infty}
\|f\|_\infty\leq1
\end{equation}
and
\begin{equation}\label{hermite}
\overline{f(\lambda)}=f(-\lambda)\qquad\text{for almost all }\lambda\in\mathbb{R}.
\end{equation}
Then there is a bivariate Gaussian process $B_t=(B^1_t,B^2_t)$ ($t\in\mathbb{R}$) with stationary increments such that
\begin{enumerate}[label={\normalfont(\roman*)}]

\item\label{hk-1} both $B^1$ and $B^2$ are two-sided Brownian motions,

\item\label{hk-2} $f$ is the cross-spectral density of $B$. That is,
\begin{equation*}
E\left[B_t^1B^2_s\right]=\frac{1}{2\pi}\int_{-\infty}^\infty\frac{(e^{-\sqrt{-1}\lambda t}-1)(e^{\sqrt{-1}\lambda s}-1)}{\lambda^2}f(\lambda)d\lambda
\end{equation*}
for any $t,s\in \mathbb{R}$.

\end{enumerate}

Conversely, if a bivariate process $B_t=(B^1_t,B^2_t)$ ($t\in\mathbb{R}$) with stationary increments satisfies condition \ref{hk-1}, there is a measurable function $f:\mathbb{R}\to\mathbb{C}$ satisfying \eqref{L_infty}--\eqref{hermite} and condition \ref{hk-2}. 
\end{proposition}
The aim of this subsection is to give a sufficient condition which makes a market driven by the process $B$ described in the above proposition free of arbitrage under market frictions. More formally, let $f$ be a measurable function satisfying \eqref{L_infty}--\eqref{hermite} and $B_t=(B^1_t,B^2_t)$ ($t\in\mathbb{R}$) be a bivariate Gaussian process with stationary increments satisfying conditions \ref{hk-1}--\ref{hk-2}. We consider the market with two risky assets where the (discounted) log-price process of the $\nu$-th risky asset is given by
\begin{equation}\label{wllag-model}
X^\nu_t=A^\nu_t+\sigma_\nu B^\nu_t,\qquad t\in[0,T]
\end{equation}
with $\sigma_\nu>0$ and $A^\nu=(A^\nu_t)_{t\in[0,T]}$ is a continuous process for each $\nu=1,2$. Hence the (discounted) price process $S=(S_t)_{t\in[0,T]}$ of the risky assets is given by $S^\nu_t=\exp(X^\nu_t)$, $t\in[0,T]$ for each $\nu=1,2$. As the filtration $\mathbb{F}$, we take the usual augmentation of the natural filtration of the process $X=(X^1,X^2)$. 
\begin{proposition}\label{wllag-cfs}
Suppose that the process $A=(A^1,A^2)$ is independent of $B$. Suppose also that 
\begin{equation}\label{eq:gsvz}
\int_{\lambda_0}^\infty\frac{\log(1-|f(\lambda)|)}{\lambda^2}d\lambda>-\infty
\end{equation}
for some $\lambda_0>0$. Then the process $X$ has CFS with respect to $\mathbb{F}$. 
\end{proposition}

\begin{rmk}
A sufficient condition for \eqref{eq:gsvz} is
\[
\limsup_{\lambda\to\infty}|f(\lambda)|<1.
\]
In the language of signal processing, $|f(\lambda)|^2$ is called the coherence and used as a measure of frequency-wise relevance between $B^1$ and $B^2$. Therefore, the above condition could be interpreted as requiring that two processes should not be perfectly correlated at high frequencies. 
\end{rmk}

Now, analogously to the previous subsection, we obtain the following result:
\begin{theorem}
Under the assumptions of Proposition \ref{wllag-cfs}, the following statements hold true:
\begin{enumerate}[label=(\alph*)]

\item The market $S$ has no arbitrage in the Cheridito class.

\item The market $S$ has an $\varepsilon$-CPS for any $\varepsilon>0$. Consequently, $S$ has no arbitrage with $\varepsilon$-transaction costs in the class $\tilde{\mathcal{A}}^\varepsilon$ for any $\varepsilon>0$. 

\end{enumerate}
\end{theorem}

As a special case, we obtain the no arbitrage properties of the model considered in \cite{HK2016,HK2017}. In these papers, to take account of potential multi-scale structures of financial markets, the following from of the cross-spectral density has been considered:
\begin{equation*}
f(\lambda)=\sum_{j=0}^\infty R_je^{-\sqrt{-1}\theta_j\lambda}1_{\Lambda_j}(\lambda),\qquad\lambda\in\mathbb{R},
\end{equation*}
where $R_j\in[-1,1]$, $\theta_j\in\mathbb{R}$ and $\Lambda_j=[-2^j\pi,-2^{j-1}\pi)\cup(2^{j-1}\pi,2^j\pi]$ for $j=0,1,\dots$. In this case, \eqref{eq:gsvz} is satisfied when $\limsup_{j\to\infty}|R_j|<1$.

\section{Proofs}\label{sec:proofs}

\subsection{Proof of Proposition \ref{hry-cfs}}

First, by Lemma \ref{lemma:augmentation} it is enough to show that $X$ has CFS with respect to $\mathbb{F}^X$. 
Moreover, by Lemma \ref{lemma:sum} we may assume $A\equiv0$ without loss of generality. 
Also, thanks to Lemma \ref{lemma:equivalence}, for the proof we may consider a particular realization of the processes $B^1$ and $B^2$, and thus we consider the realization given by \eqref{hry-real}.

Define the bivariate processes $Y=(Y_t)_{t\in[0,T]}$ and $Z=(Z_t)_{t\in[0,T]}$ by
\[
Y_t=\left(
\begin{array}{c}
Y^1_t\\
Y^2_t
\end{array}\right)
=\left(
\begin{array}{c}
\int_0^{t}\sigma_1(u)\sqrt{1-|\rho(u)|}dW^2_u\\
\int_0^{t\wedge\theta}\sigma_2(u)dW^0_{u}+\int_0^{(t-\theta)_+}\sigma_2(u)\sqrt{1-|\rho(u)|}dW^3_u
\end{array}\right)
\]
and $Z_t=X_t-Y_t$ for $t\in[0,T]$. By construction $Y$ and $Z$ are independent. Therefore, by Lemma \ref{lemma:sum} it suffices to show that $Y$ has CFS with respect to $\mathbb{F}^Y$. 

Setting 
\[
\mathcal{G}_t=\sigma(W^2_u:u\leq t)\vee\sigma(W^0_{u\wedge\theta}:u\leq t)\vee\sigma(W^3_{(u-\theta)_+}:u\leq t),\qquad t\in[0,T],
\]
we have $\mathcal{F}^Y_t\subset\mathcal{G}_t$ for every $t\in[0,T]$. Therefore, by Lemmas \ref{lemma:csbp}--\ref{lemma:large} it is enough to show that
\[
P\left(\sup_{t\in[t_0,T]}\|Y_t-Y_{t_0}-f(t)\|<\varepsilon|\mathcal{G}_{t_0}\right)>0\qquad\text{a.s.}
\]
for any $t_0\in[0,T)$, $f\in C_0([t_0,T],\mathbb{R}^2)$ and $\varepsilon>0$. Since $(Y_t-Y_{t_0})_{t\in[t_0,T]}$ is independent of $\mathcal{G}_{t_0}$, this follows from 
\[
P\left(\sup_{t\in[t_0,T]}\|Y_t-Y_{t_0}-f(t)\|<\varepsilon\right)>0.
\]
Due to the independence between $Y^1$ and $Y^2$, we obtain this inequality once we show that
\begin{equation}\label{hry-aim}
P\left(\sup_{t\in[t_0,T]}|Y^\nu_t-Y^\nu_{t_0}-f_\nu(t)|<\varepsilon\right)>0
\end{equation}
for $\nu=1,2$, where $f_\nu$ denotes the $\nu$-th component function of $f$. Moreover, we can easily see that it is sufficient to consider the case $t_0=0$. Then, \eqref{hry-aim} for $\nu=1$ immediately follows from Lemma 3.1 of \cite{Pakkanen2010}. Moreover, if $\theta\geq T$, \eqref{hry-aim} for $\nu=2$ also follows from Lemma 3.1 of \cite{Pakkanen2010} because $Y^2_t=\int_0^{t}\sigma_2(u)dW^0_{u}$ for all $t\in[0,T]$. Otherwise, noting that
\begin{multline*}
P\left(\sup_{t\in[0,T]}|Y^2_t-Y^2_{0}-f_2(t)|<\varepsilon\right)
\geq P\left(\sup_{t\in[0,\theta]}\left|\int_0^{t}\sigma_2(u)dW^0_{u}-f_2(t)\right|<\frac{\varepsilon}{2}\right)\\
\times P\left(\sup_{t\in[0,T-\theta]}\left|\int_0^{t}\sigma_2(u)\sqrt{1-|\rho(u)|}dW^3_u-\{f_2(t+\theta)-f(\theta)\}\right|<\frac{\varepsilon}{2}\right),
\end{multline*}
we obtain \eqref{hry-aim} for $\nu=2$ again by Lemma 3.1 of \cite{Pakkanen2010}.\hfill$\Box$

\subsection{Proof of Proposition \ref{wllag-cfs}}

\begin{lemma}\label{lemma:cfs-independent}
Let $Y=(Y_t)_{t\in[0,T]}$ be a $d$-dimensional continuous stochastic process such that $Y^1,\dots,Y^d$ are independent. Then, $Y$ has CFS with respect to $\mathbb{F}^Y$ if and only if $Y^\nu$ has CFS with respect to $\mathbb{F}^{Y^\nu}$ for all $\nu=1,\dots,d$.
\end{lemma}

\begin{proof}
The ``if'' part is obvious, hence we prove the ``only if'' part. By Lemma \ref{lemma:csbp} it suffices to prove
\begin{equation}\label{aim:cfs-independent}
P\left(\sup_{t\in[t_0,T]}\|Y_t-Y_{t_0}-f(t)\|<\varepsilon|\mathcal{F}^Y_{t_0}\right)>0\qquad\text{a.s.}
\end{equation}
for any $t_0\in[0,T)$, $f\in C_0([t_0,T],\mathbb{R}^d)$ and $\varepsilon>0$. For each $\nu=1,\dots,d$, we denote by $f_\nu$ the $\nu$-th coordinate function of $f$. Also, $\mathcal{F}^\nu$ denotes the $\sigma$-field generated by the process $Y^\nu$. Since $Y^1,\dots,Y^d$ are independent, we have
\begin{align*}
&P\left(\sup_{t\in[t_0,T]}\|Y_t-Y_{t_0}-f(t)\|<\varepsilon|\mathcal{F}^Y_{t_0}\right)
\geq P\left(\bigcap_{\nu=1}^d\left\{\sup_{t\in[t_0,T]}|Y^\nu_t-Y^\nu_{t_0}-f_\nu(t)|<\frac{\varepsilon}{d}\right\}|\mathcal{F}^Y_{t_0}\right)\\
&=E\left[P\left(\bigcap_{\nu=2}^d\left\{\sup_{t\in[t_0,T]}|Y^\nu_t-Y^\nu_{t_0}-f_\nu(t)|<\frac{\varepsilon}{d}\right\}|\mathcal{F}^Y_{t_0}\vee\mathcal{F}^1\right)1_{\left\{\sup_{t\in[t_0,T]}|Y^1_t-Y^1_{t_0}-f_1(t)|<\frac{\varepsilon}{d}\right\}}|\mathcal{F}^Y_{t_0}\right]\\
&=P\left(\bigcap_{\nu=2}^d\left\{\sup_{t\in[t_0,T]}|Y^\nu_t-Y^\nu_{t_0}-f_\nu(t)|<\frac{\varepsilon}{d}\right\}|\bigvee_{\nu=2}^d\mathcal{F}^{Y^\nu}_{t_0}\right)P\left(\sup_{t\in[t_0,T]}|Y^1_t-Y^1_{t_0}-f_1(t)|<\frac{\varepsilon}{d}|\mathcal{F}^{Y^1}_{t_0}\right)\\
&=\cdots=\prod_{\nu=1}^dP\left(\sup_{t\in[t_0,T]}|Y^\nu_t-Y^\nu_{t_0}-f_\nu(t)|<\frac{\varepsilon}{d}|\mathcal{F}^{Y^\nu}_{t_0}\right),
\end{align*}
hence the CFS property of $Y^\nu$ with respect to $\mathbb{F}^{Y^\nu}$ for every $\nu=1,\dots,d$ yields \eqref{aim:cfs-independent}.
\end{proof}

\begin{proof}[\upshape\bfseries Proof of Proposition \ref{wllag-cfs}]

By Lemma \ref{lemma:augmentation} it is enough to show that $X$ has CFS with respect to $\mathbb{F}^X$. 
Moreover, by Lemma \ref{lemma:sum} we may assume $A\equiv0$ without loss of generality. 

We start with constructing a realization of $B$ which is suitable to our purpose. To accomplish this, we use some concepts on Schwartz's generalized functions. We refer to Chapters 6--7 of \cite{Rudin1991} for details about them. First, define the function $g:\mathbb{R}\to\mathbb{C}$ by
\[
g(\lambda)=\left\{
\begin{array}{ll}
f(\lambda)/|f(\lambda)| & \text{if }|f(\lambda)|\neq 0,\\
0 & \text{otherwise}.
\end{array}\right.
\]
$g$ is evidently measurable. 
Next, let us denote by $\mathfrak{S}$ the set of all (complex-valued) rapidly decreasing functions on $\mathbb{R}$. Also, $L^2(\mathbb{R})$ denotes the space of all \textit{complex-valued} square-integrable functions. For a function $u\in L^2(\mathbb{R})$, $\hat{u}$ and $\check{u}$ denote the Fourier transform and the inverse Fourier transform of $u$, respectively. Here, when $u$ is integrable, $\hat{u}$ is given by $\hat{u}(\lambda)=\int_{-\infty}^\infty u(t)e^{-\sqrt{-1}\lambda t}dt$, $\lambda\in\mathbb{R}$. 
Then, we define the function $\alpha:\mathfrak{S}\to\mathbb{C}$ by $\alpha(u)=\int_{-\infty}^\infty\check{u}(\lambda)\sqrt{|f(\lambda)|}g(\lambda)d\lambda$ for $u\in\mathfrak{S}$, which can be defined thanks to \eqref{L_infty}. $\alpha$ is obviously a tempered generalized function on $\mathbb{R}$. Moreover, if $u\in\mathfrak{S}$ is real-valued, then $\alpha(u)\in\mathbb{R}$. In fact, we have
\begin{align*}
\overline{\alpha(u)}=\int_{-\infty}^\infty\overline{\check{u}(\lambda)}\sqrt{|f(\lambda)|}\cdot\overline{g(\lambda)}d\lambda
=\int_{-\infty}^\infty\check{u}(-\lambda)\sqrt{|f(-\lambda)|}g(-\lambda)d\lambda
=\alpha(u)
\end{align*}
by \eqref{hermite}. Now, for any $u\in\mathfrak{S}$ we have $\widehat{\alpha*u}=\widehat{u}\widehat{\alpha}=\widehat{u}\sqrt{|f|}g$ in $\mathfrak{S}^*$, hence $\widehat{\alpha*u}\in L^2(\mathbb{R})$. Therefore, $\alpha*u\in L^2(\mathbb{R})$ and $\|\alpha*u\|_{L^2(\mathbb{R})}=(2\pi)^{-1}\|\widehat{u}\sqrt{|f|}g\|_{L^2(\mathbb{R})}\leq\|u\|_{L^2(\mathbb{R})}$ by the Parseval identity and \eqref{L_infty}. Hence, there is a (unique) continuous function $\tilde{\alpha}:L^2(\mathbb{R})\to L^2(\mathbb{R})$ such that $\tilde{\alpha}(u)=\alpha*u$ for any $u\in\mathfrak{S}$. By continuity $\tilde{\alpha}(u)$ is real-valued as long as so is $u\in L^2(\mathbb{R})$. We also define the tempered generalized functions $\beta,\gamma$ on $\mathbb{R}$ by setting $\beta(u)=\int_{-\infty}^\infty\check{u}(\lambda)\sqrt{|f(\lambda)|}d\lambda$ and $\gamma(u)=\int_{-\infty}^\infty\check{u}(\lambda)\sqrt{1-|f(\lambda)|}d\lambda$ for $u\in\mathfrak{S}$. Then, an analogous argument to the above implies that there are continuous functions $\tilde{\beta}:L^2(\mathbb{R})\to L^2(\mathbb{R})$ and $\tilde{\gamma}:L^2(\mathbb{R})\to L^2(\mathbb{R})$ such that $\tilde{\beta}(u)=\beta*u$ and $\tilde{\gamma}(u)=\gamma*u$ for any $u\in\mathfrak{S}$ and that both $\tilde{\beta}(u)$ and $\tilde{\gamma}(u)$ are real-valued as long as so is $u\in L^2(\mathbb{R})$.

Now let $(W^k_t)_{t\in\mathbb{R}}$ ($k=1,2,3$) be three independent two-sided standard Brownian motions. Then we define the processes $(B^1_t)_{t\in\mathbb{R}}$ and $(B^2_t)_{t\in\mathbb{R}}$ by 
\[
B^1_t=
\left\{
\begin{array}{ll}
\int_{-\infty}^\infty\tilde{\alpha}(1_{(0,t]})(s)dW^1_s+\int_{-\infty}^\infty\tilde{\gamma}(1_{(0,t]})(s)dW^2_s  & \text{if }t\geq0,  \\
-\int_{-\infty}^\infty\tilde{\alpha}(1_{(0,t]})(s)dW^1_s-\int_{-\infty}^\infty\tilde{\gamma}(1_{(t,0]})(s)dW^2_s  & \text{otherwise}, 
\end{array}
\right.
\]
and
\[
B^2_t=
\left\{
\begin{array}{ll}
\int_{-\infty}^\infty\tilde{\beta}(1_{(0,t]})(s)dW^1_s+\int_{-\infty}^\infty\tilde{\gamma}(1_{(0,t]})(s)dW^3_s  & \text{if }t\geq0,  \\
-\int_{-\infty}^\infty\tilde{\beta}(1_{(0,t]})(s)dW^1_s-\int_{-\infty}^\infty\tilde{\gamma}(1_{(t,0]})(s)dW^3_s  & \text{otherwise}. 
\end{array}
\right.
\]
We show that this process $B_t=(B^1_t,B^2_t)$ $(t\in\mathbb{R})$ is a realization of the desired process. First, it is evident that both $B^1$ and $B^2$ are real-valued and Gaussian. Moreover, since it holds that $\hat{\tilde{\alpha}}(u)=\hat{u}\sqrt{|f|}g$, $\hat{\tilde{\beta}}(u)=\hat{u}\sqrt{|f|}$ and $\hat{\tilde{\gamma}}(u)=\hat{u}\sqrt{1-|f|}$ in $L^2(\mathbb{R})$ for any $u\in L^2(\mathbb{R})$, $\nu=1,2$ and $t,s\in\mathbb{R}$, we have $E[B^\nu_tB^\nu_s]=|t|\wedge|s|$ if $ts\geq0$ and $E[B^\nu_tB^\nu_s]=0$ otherwise by the Parseval identity. In particular, thanks to the Kolmogorov continuity theorem, we may assume that the process $B^\nu$ is continuous. Consequently, $B^1$ and $B^2$ satisfy condition (i). Condition (ii) also follows from the Parseval identity. This especially implies that the bivariate process $B_t=(B_t^1,B_t^2)$ is of stationary increments. Hence the process $B$ is turned out to be a realization of the desired process. 

We turn to the main body of the proof. Let us define the bivariate processes $Y=(Y_t)_{t\in[0,T]}$ and $Z=(Z_t)_{t\in[0,T]}$ by
\[
Y_t=\left(
\begin{array}{c}
Y^1_t\\
Y^2_t
\end{array}\right)
=\left(
\begin{array}{c}
\int_{-\infty}^\infty\tilde{\gamma}(1_{(0,t]})(s)dW^2_s\\
\int_{-\infty}^\infty\tilde{\gamma}(1_{(0,t]})(s)dW^3_s
\end{array}\right)
\]
and $Z_t=X_t-Y_t$ for $t\in[0,T]$. By construction $Y$ and $Z$ are independent. Moreover, we can easily check that, for each $\nu=1,2$, $Y^\nu$ is Gaussian and of stationary increments with spectral density $1-|f|$. In particular, thanks to the Kolmogorov continuity theorem, we may assume that the process $Y$ is continuous. Hence the process $Z$ is continuous as well. Therefore, by Lemma \ref{lemma:sum} it suffices to show that $Y$ has CFS with respect to $\mathbb{F}^Y$. 

By construction $Y^1$ and $Y^2$ are independent. Therefore, by Lemma \ref{lemma:cfs-independent} it is enough to show that $Y^\nu$ has CFS with respect to $\mathbb{F}^{Y^\nu}$ for each $\nu=1,2$. Since $Y^\nu$ is a continuous Gaussian process with stationary increments whose spectral density is $1-|f|$, the desired result follows from Theorem 2.1 of \cite{GSvZ2011}. This completes the proof.
\end{proof}

\section*{Acknowledgments}

We would like to thank Dacheng Xiu for asking us about the no arbitrage properties of models with lead-lag relationships. 
This work was supported by JSPS KAKENHI Grant Numbers JP16K03601, JP16K17105, JP17H01100.

{\small
\renewcommand*{\baselinestretch}{1}\selectfont
\addcontentsline{toc}{section}{References}

}

\end{document}